\documentclass[runningheads,orivec]{llncs}

\usepackage{times}
\usepackage{soul}
\usepackage{url}
\usepackage[hidelinks]{hyperref}
\usepackage[utf8]{inputenc}
\usepackage{caption}
\usepackage{booktabs}
\usepackage{amssymb}
\usepackage[normalem]{ulem}
\usepackage{MnSymbol}
\usepackage{stmaryrd}
\usepackage{graphicx}
\usepackage{amsmath}
\usepackage{bussproofs}
\newcommand{\logic}[1]{\ensuremath{\mathcal{L}_{\operatorname{#1}}}}

\newcommand{\Bigarrow}[1]{\underset{#1}{\Longrightarrow}}

\numberwithin{definition}{section}
\numberwithin{proposition}{section}
\numberwithin{theorem}{section}

\begin{document}

\mainmatter  

\title{Providing personalized Explanations: a Conversational Approach}
\titlerunning{Providing personalized Explanations: a Conversational Approach}

\author{Jieting Luo\inst{1}, Thomas Studer \inst{2} \and Mehdi Dastani \inst{3}}
\authorrunning{Jieting Luo et al.} 
%
\tocauthor{Jieting Luo, Thomas Studer and Mehdi Dastani}
\institute{Zhejiang University, Hangzhou, China\\
\email{luojieting@zju.edu.cn}
\and
University of Bern, Bern, Switzerland\\
\email{thomas.studer@unibe.ch}
\and
Utrecht University, Utrecht, the Netherlands\\
\email{M.M.Dastani@uu.nl}
}

\maketitle
\bibliographystyle{splncs04}

\begin{abstract}
The increasing applications of AI systems require personalized explanations for their behaviors to various stakeholders since the stakeholders may have various knowledge and backgrounds. In general, a conversation between explainers and explainees not only allows explainers to obtain the explainees' background, but also allows explainees to better understand the explanations. In this paper, we propose an approach for an explainer to communicate personalized explanations to an explainee through having consecutive conversations with the explainee. We prove that the conversation terminates due to the explainee's justification of the initial claim as long as there exists an explanation for the initial claim that the explainee understands and the explainer is aware of. 
\end{abstract}

\keywords{Explanation; Personalization; Conversation; Justification Logic; Modal Logic}

\section{Introduction}

Explainable artificial intelligence (XAI) is one of the important topics in artificial intelligence due to the recognition that it is important for humans to understand the decisions or predictions made by the AI \cite{miller2019explanation}. Understanding the behavior of AI systems does not only improve the user experience and trust in such systems, but it also allows engineers to better configure them when their behavior needs to change. This is particularly important when AI systems are used in safety-critical domains such as healthcare where decisions made or influenced by clinical decision support systems ultimately affect human life and well-being. In such systems, related stakeholders and professionals must understand how and why certain decisions are made \cite{antoniadi2021current}. 

An important requirement for explainable AI systems is to ensure that the stakeholders with various backgrounds understand the provided explanations and the underlying rationale and inner logic of the decision or prediction process~\cite{tsai2021effects}. For example, an explanation of why we should drink enough water throughout the day that is formulated in specialized medical terms may be understandable to medical professionals, but not to young children who have no medical knowledge or background. Therefore, AI systems should be able to provide \emph{personalized} and relevant explanations that match the knowledge and background of their stakeholders. Hilton~\cite{hilton1990conversational} argues that an explanation is a \emph{social} process of conveying why a claim is made to someone. It is the conversations between explainers and explainees that allow explainers to obtain the explainees' background and allow explainees to better understand the explanations. 

In this paper, we propose a novel approach to automatically construct and communicate personalized explanations based on the explainer's beliefs about the explainee's background. Our approach is based on conversations between an explainer and an explainee, and exploits tools and results from justification logic~\cite{artemovFittingBook,justificationLogic2019}.
The first justification logic, the Logic of Proofs, has been introduced by Artemov to give a classical provability interpretation to intuitionistic logic~\cite{Art01BSL}. Later, various possible worlds semantics for justification logic have been developed~\cite{Art12SL,Fit05APAL,KuzStu12AiML,StuderLehmannSubsetModel2019}, which led to epistemic interpretations of justification logic. We will use a logic that features both the modal $\Box$-operator and explicit justification terms $t$, a combination that goes back to~\cite{ArtNog05JLC}. Our approach is built on the idea of reading \emph{an agent understands an explanation~$\mathcal{E}$ for a claim $F$} as \emph{the agent has a justification $t$ for $F$ such that $t$ also justifies all parts of $\mathcal{E}$}. This is similar to the logic of knowing why~\cite{KnowingWhy} where \emph{knowing why $F$} is related to \emph{having a justification for $F$}. With this idea, the explainer can interpret the explainee's background from his feedback and provide further explanations given what he has learned about the explainee.

We first develop a multi-agent modular model that allows us to represent and reason about agents' beliefs and justification. We then model how the explainee gains more justified beliefs from explanations, and how the explainer specifies his preferences over available explanations using specific principles. We finally model the conversation where the explainee provides his feedback on the received explanation and the explainer constructs a further explanation given his current beliefs about the explainee's background interpreted from the explainee's feedback. Our approach ensures that the conversation will terminate due to the explainee's justification of the initial claim as long as there exists an explanation for the initial claim that the explainee understands and the explainer is aware of.

\section{Multi-agent Modular Models}
Let $\mathrm{Prop}$ be a countable set of atomic propositions. The set of propositional formulas~\logic{Prop} is inductively defined as usual from $\mathrm{Prop}$, the constant $\bot$, and the binary connective $\to$. We now specify how we represent justifications and what operations on them we consider. We assume a countable set $\mathrm{JConst} = \{c_0, c_1, \ldots\}$ of justification constants. Further, we assume a countable set $\mathrm{JVar}$ of justification variables, where each variable is indexed by a propositional formula and a (possibly empty) list of propositional formulas, i.e.~if $A_1,\ldots,A_n,B \in \logic{Prop}$, then $x_B^{A_1,\ldots,A_n}$ is a justification variable.
Constants denote atomic justifications that the system no longer analyzes, and variables denote unknown justifications. Justification terms are defined inductively as follows:
\[t ::= c \mid x \mid t \cdot t \]
where $c \in \mathrm{JConst}$ and $x \in \mathrm{JVar}$. We denote the set of all terms by $\mathrm{Tm}$. A term is ground if it does not contain variables, so we denote the set of all ground terms by $\mathrm{Gt}$. A term can be understood as a proof or an evidence. Let $Agt$ be a finite set of agents. Formulas of the language $\logic{J}$ are defined inductively as follows:
\[A ::= p \mid \bot \mid A \to A \mid \Box_i A  \mid \llbracket t \rrbracket_i A \]
where $p \in \mathrm{Prop}$, $i \in Agt$ and $t \in \mathrm{Tm}$. Formula $\Box_i A$ is interpreted as ``agent $i$ believes $A$'', and formula $\llbracket t \rrbracket_i A$ is interpreted as ``agent $i$ uses $t$ to justify $A$''.

The model we use in this paper is a multi-agent modular model that interprets justification logic in a multi-agent context. It is a Kripke frame extended with a set of agents and two evidence functions. In general, evidence accepted by different agents are distinct, so evidence terms for each agent are constructed using his own basic evidence function.
\begin{definition}[Multi-agent Modular Models]\label{model}
A multi-agent modular model over a set of atomic propositions $\mathrm{Prop}$ and a set of terms $\mathrm{Tm}$ is defined as a tuple $\mathcal{M} = (Agt, W, \tilde{R}, \tilde{*}, \pi)$, where
\begin{itemize}
\item $Agt = \{1, 2\}$ is a set of agents, we assume that it is always the case that agent 1 announces an explanation to agent 2; 
\item $W \not= \emptyset$ is a set of worlds;
\item $\tilde{R} = \{R_1, R_2\}$ for each agent in $Agt$, where $R_i \subseteq W \times W$ is a reflexive and transitive accessibility relation;
\item $\tilde{*} = \{*_1, *_2\}$ for each agent in $Agt$, where $*_i$ is an evidence function $*_i: \mathrm{Tm} \times W \to \mathcal{P}(\logic{Prop})$ that maps a term $t \in \mathrm{Tm}$ and a world $w \in W$ to a set of formulas in $\logic{Prop}$; 
\item $\pi: \mathrm{Prop} \to 2^W$ is an evaluation function for the interpretation of propositions.
\end{itemize}
\end{definition}
We assume that the agents have finite reasoning power. Therefore, we restrict our models such that for each $w\in W$ and agent $i$,
\begin{itemize}
\item there are only finitely many $t \in \mathrm{Gt}$ such that $*_i(t,w)$ is non-empty, and
\item for each $t \in \mathrm{Gt}$, the set  $*_i(t,w)$ is finite.
\end{itemize}
Moreover, for any agent $i$ and world $w$, it is not necessary that 
\[
\text{
$(F \to G) \in *_i(s, w)$ and $F \in *_i(t, w)$ imply $G \in *_i(s \cdot t, w)$
}\tag{\dagger}
\]

\begin{definition}[Truth Evaluation]
We define what it means for a formula $A$ to hold under a multi-agent modular model $\mathcal{M}$ and a world $w$, written as $\mathcal{M}, w \models A$, inductively as follows:
\begin{itemize}
\item $\mathcal{M}, w \not \models \bot$;
\item $\mathcal{M}, w \models P$ iff $w \in \pi(P)$;
\item $\mathcal{M}, w \models F \to G$ iff  $\mathcal{M}, w \not \models F$ or $\mathcal{M}, w \models G$;
\item $\mathcal{M}, w \models \Box_i F$ iff for any $u \in W$, if $w R_i u$, then $\mathcal{M}, u \models F$;
\item $\mathcal{M}, w \models \llbracket t \rrbracket_i F$ iff $F \in *_i(t, w)$.
\end{itemize}
Other classical logic connectives (e.g.,``$\land$'', ``$\lor$'') are assumed to be defined as abbreviations by using $\bot$ and $\to$ in the conventional manner. We say that a formula $A$ is valid in a model $\mathcal{M}$, written as $\mathcal{M} \models A$ if $\mathcal{M}, w \models A$ for all $w \in W$. We say that a formula $A$ is a validity, written as $\models A$ if $\mathcal{M} \models A$ for all models $\mathcal{M}$. 
\end{definition}

We require that modular models satisfy the property of \emph{justification yields belief}:
for any ground term $t$, agent $i$, and world $w$, if $F \in *_i(t, w)$, then for any $u \in W$, if $w R_i u$, then $\mathcal{M}, u \models F$, which gives rise to the following validity:
\begin{equation}\tag{JYB}
\models \llbracket t \rrbracket_i F \to  \Box_i F.
\end{equation}
Note that in contrast to usual models of justification logic, we require justification yields belief only for ground terms (and not for all terms as is originally required in modular models). The reason is that we interpret justification variables in a new way. Traditionally, a justification variable stands for an arbitrary justification. Hence   $\llbracket x \rrbracket_i F \to  \Box_i F$ should hold: no matter which justification we have for $F$, it should yield belief of $F$. In this paper, we use a different reading of justification variables. They stand for open assumptions, which do not (yet) have a justification. Therefore, $\llbracket x \rrbracket_i F$ will not imply belief of $F$. Our modular model gives rise to the following validity due to the reflexivity of the accessibility relations:
\[ \models \Box_i F \to F. \]
Combining this with (JYB), we find that justifications by ground terms are factive: for any ground term $t$ and any formula $F$, we have
\[\models \llbracket t \rrbracket_i F \to F.\]
Notice that our model does not respect the usual application operation ($\cdot$) on evidence terms due to the removal of constraint ($\dagger$) from our model, 
\[
\not\models \llbracket s \rrbracket_i (F \to G) \to (\llbracket t \rrbracket_i F \to \llbracket s \cdot t \rrbracket_i G),
\]
because agents are limited in their reasoning powers and thus might not be able to derive all of the logical consequences from their justified beliefs by constructing proofs, which becomes the reason why agents need explanations.

\section{Understanding and Learning from Explanations}
Given a claim, an agent can construct a deduction for a claim, which is what we call an explanation of the claim in this paper. An explanation is inductively defined as a tree of formulas.
\begin{definition}[Explanations]
Given formulas $A_1, \ldots, A_n, B \in \logic{Prop}$, a simple explanation is of the form
\begin{prooftree} 
\AxiomC{$A_1, \ldots, A_n$}
\UnaryInfC{$B$}
\end{prooftree}
An explanation, denoted as $\mathcal{E}$, is inductively defined as follows: it is a simple explanation or of the form
\begin{prooftree} 
\AxiomC{$\mathcal{E}_1, \ldots, \mathcal{E}_n$}
\UnaryInfC{$B$}
\end{prooftree}
where $\mathcal{E}_1, \ldots, \mathcal{E}_n$ are explanations. We say that 
\begin{itemize}
    \item $B$ is the claim of $\mathcal{E}$, denoted as $claim(\mathcal{E}) = B$;
    \item  $Pr(\mathcal{E}, B)$ is the list of premises of $B$ in $\mathcal{E}$, that is,  $Pr(\mathcal{E}, B) = A_1, \ldots, A_n$ if $\mathcal{E}$ is a simple explanation; otherwise, $Pr(\mathcal{E}, B) = claim(\mathcal{E}_1), \ldots, claim(\mathcal{E}_n)$;
    \item formula $F$ in $\mathcal{E}$ is a hypothesis if $Pr(\mathcal{E}, F) = \emptyset$, and $H(\mathcal{E})$ is the set of hypotheses of $\mathcal{E}$; 
    \item formula $F$ in $\mathcal{E}$ is a derived formula if $Pr(\mathcal{E}, F) \not= \emptyset$, and $D(\mathcal{E})$ is the set of derived formulas of $\mathcal{E}$. 
\end{itemize}
\end{definition}


One important property of justification logic is its ability to internalize its own notion of proof. 
If $B$ is derivable from $A_1, \ldots, A_n$, then there exists a term $t \cdot x_1 \cdot \cdots \cdot x_n $ such that $\llbracket t \cdot x_1 \cdot \cdots \cdot x_n  \rrbracket_i B$ is derivable from $\llbracket x_1 \rrbracket_i A_1, \ldots, \llbracket x_n \rrbracket_i A_n$.\footnote{This property requires a so-called axiomatically appropriate constant specification.} The justification term $t \cdot x_1 \cdot \cdots \cdot x_n$ justifying $B$ represents a blueprint of the derivation of $B$ from $A_1, \ldots, A_n$. In this section, we will define a procedure that can mimic the application operation on terms to construct derived terms in order to internalize the deduction of an explanation. Given an explanation, a derived term of the conclusion with respect to the explanation is constructed with the justifications of the premises and the deduction from the premises to the conclusion. Typically, if there exists any premises that the agent cannot justify, then a variable is used for its justification in the derived term; if the agent cannot justify the deduction, then a variable is used as the derived term.
\begin{definition}[Construction of Derived Terms]
Given a multi-agent modular model $\mathcal{M}$, a world $w$, an explanation $\mathcal{E}$, and a derived formula $B$ occurring in $\mathcal{E}$, we define agent 2's derived term of $B$ with respect to $\mathcal{E}$ inductively as follows:
\begin{itemize}
\item Case: $B$ is the claim of a simple explanation $\mathcal{E'} = A_1, \ldots, A_n / B$.  We distinguish two cases:
\begin{enumerate}
\item
If  there exists $d \in \mathrm{Gt}$ such that $\mathcal{M}, w \models \llbracket d \rrbracket_2 (A_1 \to  (\cdots \to (A_n \to B)\cdots))$, then the derived term has the form $d \cdot t_1 \cdots t_n$ where
the terms $t_i$ are given by:
if there exists $s_i \in \mathrm{Gt}$ with $\mathcal{M}, w \models \llbracket s_i \rrbracket_2 A_i$, then set $t_i = s_i$; else, set  $t_i = x_{A_i}$;
\item
otherwise, the derived term of $B$ has the form  $x_B^{Pr(\mathcal{E}, B)}$.
\end{enumerate}
\item Case: $B$ is the claim of an explanation $\mathcal{E'} = \mathcal{E}'_1, \ldots \mathcal{E}'_n / B$. 
We distinguish two cases:
\begin{enumerate}
\item
If  there exists $d \in \mathrm{Gt}$ such that 
\[
\mathcal{M}, w \models \llbracket d \rrbracket_2 (claim(\mathcal{E}'_1) \to  (\cdots \to (claim(\mathcal{E}'_n) \to B)\cdots)), 
\]
then the derived term has the form $d \cdot t_1 \cdots t_n$ where
each $t_i$ is the derived term of $claim(\mathcal{E}'_i)$ with respect to $\mathcal{E}$;
\item
otherwise, the derived term of $B$ has the form  $x_B^{Pr(\mathcal{E}, B)}$.
\end{enumerate}
\end{itemize}
\end{definition}

\begin{example}
Assume that we have a multi-agent modular model $\mathcal{M}$. Agent~2 hears an example $\mathcal{E} = A / B / C$ in world $w$, and it is the case that
\begin{align*}
& \mathcal{M}, w \models \llbracket t_A \rrbracket_2 A\\
& \mathcal{M}, w \models \llbracket d_{A \to B} \rrbracket_2 (A \to B)\\
& \mathcal{M}, w \models \llbracket d_{B \to C} \rrbracket_2 (B \to C)
\end{align*}
for ground terms $t_A$, $d_{A \to B}$, and $d_{B \to C}$, then the derived term of $B$ with respect to $\mathcal{E}$ is $d_{A \to B} \cdot t_A$, and the derived term of $C$ with respect to $\mathcal{E}$ is $d_{B \to C} \cdot (d_{A \to B} \cdot t_A)$. If agent 2 cannot justify $A$, then the derived term of $C$ with respect to $\mathcal{E}$ would become $d_{B \to C} \cdot (d_{A \to B} \cdot x_A)$; if agent 2 cannot justify $A \to B$, then the derived term of $C$ with respect to $\mathcal{E}$ would become $d_{B \to C} \cdot x_B^A$.
\end{example}
Agent 2's justification for a deduction can be seen as his reasoning capability and can be different from agent to agent. If agent 2 cannot justify a deduction step, then the deduction is beyond his reasoning capability. In real life, agents' reasoning capabilities can be limited by factors such as age, profession and experience. For example, a mathematician can follow complicated mathematical proofs, while a primary student can only follow simple mathematical proofs. Further, for a derived formula in an explanation, agent 2 might have another term that has nothing to do with the explanation to justify it. But using this term to justify the formula does not mean that agent 2 can follow the explanation, so we need to require that a derived term be formed by justification terms that are used to justify its premises and deduction in the explanation. We should also notice that a derived term of a derived formula with respect to an explanation might not be unique, because there might exist multiple terms for agent 2 to justify the hypotheses in the explanation, making the derived terms different. Intuitively, an agent understands an explanation if the derived term of its conclusion does not contain any variables (unknown justification), i.e., it is a ground term. 
\begin{definition}[Understanding Explanations]
Given a multi-agent modular model~$\mathcal{M}$, a world $w$ and an explanation $\mathcal{E}$, let $t$ be agent 2's derived term of $claim(\mathcal{E})$ with respect to $\mathcal{E}$ in world $w$. We say that agent 2 understands~$\mathcal{E}$ in world $w$ iff $t$ is a ground term.
\end{definition}
Thus, if derived term $t$ contains variables, meaning that there exists a hypothesis or a deduction from $\mathcal{E}$ that agent 2 cannot justify, then agent 2 cannot understand $\mathcal{E}$.
\begin{example}
In Example 1, agent 2 understands explanation $\mathcal{E}$ if the derived term of $C$ with respect to $\mathcal{E}$ is $d_{B \to C} \cdot (d_{A \to B} \cdot t_A)$; agent 2 cannot understand explanation $\mathcal{E}$ if the derived term of $C$ with respect to $\mathcal{E}$ is $d_{B \to C} \cdot (d_{A \to B} \cdot x_A)$ or $d_{B \to C} \cdot x_B^A$.
\end{example}

Once an agent hears an explanation, he can update his justification with derived terms that he constructs, which means that the agent learns from the explanation and has more justified beliefs.
\begin{definition}[Learning from Explanations]\label{def:learning}
Given a multi-agent modular model $\mathcal{M}$, a world $w$ and an explanation $\mathcal{E}$, after agent 2 hears $\mathcal{E}$ in $w$, $\mathcal{M}$ is updated as $\mathcal{M}|(2, w, \mathcal{E})$, where $\mathcal{M}|(2, w, \mathcal{E}) = (Agt, W, \tilde{R}, \tilde{*^\prime}, \pi)$ is defined as follows:  
\begin{itemize}
    \item for any $F \in D(\mathcal{E})$, $*_2^\prime(t, w) = *_2(t, w) \cup \{F\}$, where $t$ is agent 2's derived term of $F$ with respect to $\mathcal{E}$;
    \item for any $s \in \mathrm{Tm}$ and any $G \in \logic{Prop}$, if $G \in *_2(s, w)$, then $G \in *_2^\prime(s[r / x_{claim(\mathcal{E})}], w)$ and $G \in *_2^\prime(s[r / x_{claim(\mathcal{E})}^{H(\mathcal{E})}], w)$, where $r$ is agent 2's derived term of $claim(\mathcal{E})$ with respect to $\mathcal{E}$;
    \item for agent 1, $*_1^\prime(\cdot) = *_1(\cdot)$.
\end{itemize}
\end{definition}
In words, after agent 2 hears explanation $\mathcal{E}$ in world $w$, for each derived formula $F$ in $\mathcal{E}$, agent 2's justification of $F$ will be updated with its derived term $t$; the derived term $r$ of $claim(\mathcal{E})$ with respect to $\mathcal{E}$ is substituted for every occurrence of $x_{claim(\mathcal{E})}$ and $x_{claim(\mathcal{E})}^{H(\mathcal{E})}$ in agent 2's justification. Recall that agents in this paper have limited reasoning powers and thus might not be able to derive all of the logical consequences from their beliefs. Hearing an explanation allows agent 2 to gain new justified beliefs by connecting existing justified beliefs. Note that we do not need to remove agent 2's epistemic access $R_2$ to worlds where $F$ does not hold in order to guarantee JYB. The reason is as follows: if $t$ is not a ground term, then agent 2 has yet to justify $F$ and thus agent 2's belief should remain the same as before learning from explanation $\mathcal{E}$; if $t$ is a ground term, then agent 2 believes $F$ due to the JYB constraint, and the belief of $F$ is ensured by the validity of the modal $k$-axiom (our model still respects the epistemic closure). Other agents' justification and epistemic access remain the same. As standard, the resulting updated model is still a multi-agent modular model.
\begin{proposition}
Given a multi-agent modular model $\mathcal{M}$, a world $w$ and an explanation $\mathcal{E}$, after agent 2 hears explanation $\mathcal{E}$ in world $w$, $\mathcal{M}$ is updated as $\mathcal{M}|(2, w, \mathcal{E})$, which is still a multi-agent modular model.
\end{proposition}

We then extend our language $\logic{J}$ with a new formula of the form $[i: \mathcal{E}]\varphi$, read as ``$\varphi$ is true after agent $i$ hears explanation $\mathcal{E}$'', and its evaluation is defined with respect to a multi-agent modular model $\mathcal{M}$ and a world $w$ as follows:
\[
\mathcal{M}, w \models [i: \mathcal{E}]\varphi \quad \text{iff} \quad \mathcal{M}|(i, w, \mathcal{E}), w \models \varphi.
\]
The learning process gives rise to some intuitive consequences. First of all, when agent 2 is already aware of the explanation before it is announced, agent 2 will not gain any new justified from the explanation. Secondly, it is possible for agent 2 to gain justification for the formulas that are not contained in the explanation, because the learning process contains substituting derived terms of formulas for their corresponding variables, which means that agent 2 can justify more than what an explanation has. 

Since an explanation is defined as a formula tree, we have the sufficient and necessary conditions for understanding an explanation: an agent understands an explanation if and only if the agent can justify all the hypotheses and deduction steps that are used in the explanation. Let $G$ be a formula and $L=A_1,\ldots,A_n$ be a list of formulas. Then the expression $\Bigarrow{L}  G$ stands for $A_1 \to (\cdots \to (A_n \to G)\cdots)$.
\begin{proposition}\label{prop:understanding:1}
Given a multi-agent modular model $\mathcal{M}$, a world $w$ and an explanation $\mathcal{E}$, agent 2 understands explanation $\mathcal{E}$ iff
\begin{itemize}
\item for any $F \in H(\mathcal{E})$, there exists a ground term $t \in \mathrm{Gt}$ such that 
$
\mathcal{M}, w \models \llbracket t \rrbracket_2 F,
$
and
\item for any $G \in D(\mathcal{E})$, there exists a ground term $s \in \mathrm{Gt}$ such that 
\[
\mathcal{M}, w \models \llbracket s \rrbracket_2 (\Bigarrow{Pr(\mathcal{E}, G)}  G
),
\]
\end{itemize}
\end{proposition}
Conversely, if there exists a hypothesis that the agent cannot justify, or if there exists a deduction step that is beyond the agent's reasoning capability, the agent cannot understand the announced explanation.

Another important property of understanding and learning from an explanation is that for an agent's justified beliefs that have nothing to do with the explanation, it will remain the same after hearing an explanation, and it was also the case before hearing the explanation.
\begin{proposition}
Given a multi-agent modular model $\mathcal{M}$, a world $w$ and an explanation $\mathcal{E}$, if agent 2 constructs a derived term $t_F$ for each derived formula $F$ with respect to $\mathcal{E}$, then for any ground term $s \in \mathrm{Gt}$ that does not contain $t_F$ and any formula $P \in \logic{Prop}$,
\[\mathcal{M}, w \models \llbracket s \rrbracket_2 P \leftrightarrow [2: \mathcal{E}] \llbracket s \rrbracket_2 P.\]
\end{proposition}
The direction from left to right is the persistence principle saying that an agent's beliefs always get expanded after hearing an explanation. The direction from right to left states that if after hearing the explanation the agent believes $P$ for a reason that is independent on the explanation, then before hearing the explanation the agent already believed $P$ for the same reason. In other words, having terms in our logical language also allows the agent to distinguish the justified beliefs due to the explanation from the justified beliefs due to another, unrelated reasons, which purely modal logic cannot formulate.
\begin{example}
Continued with Example 1, suppose agent 2 uses term $d_{B \to C} \cdot x_B^A$ to justify $C$ because he cannot justify the deduction from $A$ to $B$. After hearing explanation $\mathcal{E}$, it is the case that
\[
\mathcal{M}, w \models [2: \mathcal{E}] \llbracket d_{B \to C} \cdot x_B^A \rrbracket_2 C.
\]
Agent 2 continues to hear another explanation $\mathcal{E}^\prime = A / D / B$, and it is the case that
\begin{align*}
&\mathcal{M}, w \models [2: \mathcal{E}]\llbracket t_A \rrbracket_2 A,  \\ 
&\mathcal{M}, w \models [2: \mathcal{E}]\llbracket d_{A \to D} \rrbracket_2 (A \to D), \\
&\mathcal{M}, w \models [2: \mathcal{E}]\llbracket d_{D \to B} \rrbracket_2 (D \to B).
\end{align*}
Agent 2 then can construct the derived term of $B$ with respect to $\mathcal{E}^\prime$ as $d_{D \to B} \cdot (d_{A \to D} \cdot t_A)$, 
\[
\mathcal{M}, w \models [2: \mathcal{E}^\prime][2: \mathcal{E}] \llbracket d_{D \to B} \cdot (d_{A \to D} \cdot t_A) \rrbracket_2 B.
\]
Moreover, according to the learning approach in Definition \ref{def:learning}, $d_{D \to B} \cdot (d_{A \to D} \cdot t_A)$ is substituted for every occurrence of $x_B^A$ in agent 2's justification. Thus, the derived term of $C$ with respect to $\mathcal{E}$ becomes $d_{B \to C} \cdot (d_{D \to B} \cdot (d_{A \to D} \cdot t_A))$, 
\[
\mathcal{M}, w \models [2: \mathcal{E}^\prime][2: \mathcal{E}] \llbracket d_{B \to C} \cdot (d_{D \to B} \cdot (d_{A \to D} \cdot t_A)) \rrbracket_2 C.
\]
\end{example}

Now we can summarize that a user-agent profile consists of his justified beliefs and deductions. It is important for agent 1 to gain information about these two aspects of agent 2 through having consecutive conversations with agent 2 in order to provide an explanation that can be understood by agent 2.

\section{Explanation Evaluation}
In the previous section, we presented how agent 2's mental state is updated after hearing an explanation. In this section, we will investigate how agent 1 selects an explanation for announcing to agent 2. First of all, an explanation is selected from the explanations that the explainer agent is aware of, namely the explanations where all the hypotheses as well as all the derived formulas are justified by the explainer agent with ground derived terms with respect to the explanation. Secondly, an explanation that contains information that the explainee agent cannot justify should not be selected. In order to express these two requirements, we need to extend our language $\logic{J}$ with a new formula of the form $\triangle_i P$, read as "agent $i$ can justify $P$", and its evaluation is defined with respect to a multi-agent modular model $\mathcal{M}$ and a world $w$ as follows:
\begin{itemize}
\item $\mathcal{M}, w \models \triangle_i P$ iff there exists $t \in \mathrm{Gt}$ such that $\mathcal{M}, w \models \llbracket t \rrbracket_i P$.
\end{itemize}
Compared with formula $\llbracket t \rrbracket_i P$, the term $t$ is omitted in formula $\triangle_i P$, meaning that agent~$i$ can justify $A$ but we don't care how he justifies $P$. Because of the JYB constraint, we have the following validity:
\[ \models \triangle_i P \to \Box_i P.\]
\begin{definition}[Available Explanations]
Given a multi-agent modular model $\mathcal{M}$ and a world $w$, we say that an explanation $\mathcal{E}$ is available for agent 1 to agent 2 iff
\begin{itemize}
    \item for any $P \in H(\mathcal{E})$, there exists $t \in \mathrm{Gt}$ such that $\mathcal{M}, w \models \llbracket t \rrbracket_1 P$;
    \item for any $Q \in D(\mathcal{E})$, there exists $s \in \mathrm{Gt}$ such that $s$ is a derived term of $Q$ with respect to $\mathcal{E}$ and $\mathcal{M}, w \models \llbracket s \rrbracket_1 Q$;
    \item there does not exist $P \in H(\mathcal{E})$ such that $\mathcal{M}, w \models \Box_1 \lnot \triangle_2 P$;
    \item there does not exist $Q \in D(\mathcal{E})$ such that $\mathcal{M}, w \models \Box_1 \lnot \triangle_2 (    \Bigarrow{Pr(\mathcal{E}, Q)} Q)$.
\end{itemize}
%
Given a formula $F$ as a claim and a set of formulas $A$ as hypotheses, the set of agent~1's available explanations to agent 2 for proving $F$ from $A$ is denoted as 
\begin{multline*}
\lambda^{\mathcal{M},w}_{1, 2}(A, F) = \{\mathcal{E} \mid claim(\mathcal{E}) = F, H(\mathcal{E}) = A \text{ if } A \not= \emptyset, \text{ and }\\
\text{$\mathcal{E}$ is available for agent 1 to agent 2 given $\mathcal{M}$ and $w$}\}.
\end{multline*}
We  write $\lambda^{\mathcal{M},w}(A, F)$ when the agents are clear from the context.

%
\end{definition}
Compared with agent 2 that needs to construct and learn derived terms of derived formulas in an explanation, agent 1 has already justified all the derived formulas in an explanation that is available to him by derived terms and makes sure that there does not exist a hypothesis or a deduction that agent 2 cannot justify.

If an explanation is available to agent 1, he can announce it to agent 2. But when there are multiple available explanations, agent 1 must select one among them given what he believes about agent 2. The question is what principle agent 1 can hold for explanation selection. First of all, agent 1 should select an explanation that is most possible to be understood by agent 2. Looking back at our definition of understanding an explanation, we can say that one explanation is more possible to be understood by an agent than another explanation if the former one contains more hypotheses and deductions that are justified by the agent than the latter one. Besides, agent~1 is supposed to make simple explanations, which means that explanations that contain fewer deduction steps (i.e., derived formulas) are more preferred. Since the goal of agent 1 is to announce an explanation that can be understood by agent 2, it makes sense for the first principle to have priority from the second one. In this paper, we impose a total pre-order $\precsim$ over available explanations to represent the preference between two explanations with respect to these two principles. For better expression, we use $N^{\mathcal{M},w}_{1,2}(\mathcal{E})$ to denote the set of hypotheses and deductions in explanation $\mathcal{E}$ that agent 1 is not sure whether agent 2 can justify or not.
\begin{multline*}
N^{\mathcal{M},w}_{1,2}(\mathcal{E}) = \{F \mid \mathcal{M}, w \models \lnot \Box_1 \triangle_2 F, \text{ where $F$ is a hypothesis in $\mathcal{E}$, or} \\
\mathcal{M}, w \models \lnot \Box_1 \triangle_2 (\Bigarrow{Pr(\mathcal{E}, F)}  F), \text{ where $F$ is a derived formula in $\mathcal{E}$}.\} 
\end{multline*}
We might write $N^{\mathcal{M},w}(\mathcal{E})$ for short if the agents are clear from the context.

\begin{definition}[Preferences over Available Explanations]\label{def:principles}
Given a multi-agent modular model $\mathcal{M}$ and a world $w$, agent 1 provides explanations to agent 2, for any two explanations $\mathcal{E}, \mathcal{E}^\prime$, $\mathcal{E} \precsim^{\mathcal{M}, w} \mathcal{E}^\prime$ iff
\begin{itemize}
\item $|N^{\mathcal{M},w}(\mathcal{E})| > |N^{\mathcal{M},w}(\mathcal{E}^\prime)|$; or
\item $|N^{\mathcal{M},w}(\mathcal{E})| = |N^{\mathcal{M},w}(\mathcal{E}^\prime)|$ and $D(\mathcal{E}) \geq D(\mathcal{E}^\prime)$.
\end{itemize}
\end{definition}
As is standard, we also define $\mathcal{E} \sim^{\mathcal{M}, w} \mathcal{E}^\prime$ to mean $\mathcal{E} \precsim^{\mathcal{M}, w} \mathcal{E}^\prime$ and $\mathcal{E}^\prime \precsim^{\mathcal{M}, w} \mathcal{E}$, and $\mathcal{E} \prec^{\mathcal{M}, w} \mathcal{E}^\prime$ to mean $\mathcal{E} \precsim^{\mathcal{M}, w} \mathcal{E}^\prime$ and $\mathcal{E} \not \sim^{\mathcal{M}, w} \mathcal{E}^\prime$. The above definition of the preference between two explanations specifies how agent 1 selects an explanation to announce to agent 2: given two available explanations $\mathcal{E}$ and $\mathcal{E}^\prime$, agent 1 first compares two explanations in terms of the number of hypotheses and deductions that might not be justified by agent 1, and the one with less number is more preferable; if both explanations have the same number of hypotheses and deductions that might not be justified by agent 2, then agent 1 compares these two explanations in terms of the number of deduction steps in the explanations, and the one with less number is more preferable. Using this approach, agent 1 always cut out the part that he knows for sure that agent~2 can justify, or replace some part of the explanation with a shorter deduction that he knows that agent 2 can justify, making explanations shorter.

\begin{example}
Assume that we have a multi-agent modular model $\mathcal{M}$. In world $w$, agent 1 has two available explanations $\mathcal{E}_1$ and $\mathcal{E}_2$ to agent 2, where $\mathcal{E}_1 = A / B$ and $\mathcal{E}_2 = A / C / B$. Agent 1 believes that agent 2 can justify the deductions from $A$ to $B$, from $A$ to $C$, and from $C$ to $B$, but agent 1 is not sure whether agent 2 can justify the hypothesis $A$. In this case, the numbers of hypotheses and deductions in $\mathcal{E}_1$ and $\mathcal{E}_2$ that agent 1 is not sure whether agent 2 can justify are the same, namely $N(\mathcal{E}_1)^{\mathcal{M}, w} = N(\mathcal{E}_2)^{\mathcal{M}, w}$, so agent 1 needs to compare the numbers of deduction steps in $\mathcal{E}_1$ and $\mathcal{E}_2$. Because $\mathcal{E}_1$ is shorter than $\mathcal{E}_2$, namely $D(\mathcal{E}_1) < D(\mathcal{E}_2)$, we have $\mathcal{E}_1 \succ^{\mathcal{M}, w} \mathcal{E}_2$.
\end{example}

\section{Conversational Explanations}
Agent 1 has incomplete information about agent 2 in terms of his justified beliefs, but agent 1 can gain more and more information through having feedback from agent 2 on the explanations that agent 2 has announced. We first define agent 2' feedback. After agent 2 hears an explanation, he can evaluate whether he can understand the explanation. So his feedback is defined inductively as a tree that is isomorphic to a given explanation so that each node in the feedback tree corresponds to a specific formula in the explanation.
\begin{definition}[Explainees' Feedback]
Given an explanation $\mathcal{E}$, agent 2's feedback on explanation $\mathcal{E}$, denoted as $\mathcal{F}_2(\mathcal{E})$, is defined as follows:
\begin{itemize}
\item if $\mathcal{E} = A_1, \ldots A_n / B$, then $\mathcal{F}_2(\mathcal{E})$ is of the form
\begin{prooftree} 
\AxiomC{$f_1, \ldots, f_n$}
\UnaryInfC{$f_B$}
\end{prooftree}
\item if $\mathcal{E} = \mathcal{E}_1, \ldots \mathcal{E}_n / B$, then $\mathcal{F}_2(\mathcal{E})$ is of the form
\begin{prooftree} 
\AxiomC{$\mathcal{F}_2(\mathcal{E}_1), \ldots, \mathcal{F}_2(\mathcal{E}_n)$}
\UnaryInfC{$f_B$}
\end{prooftree}
\end{itemize}
where $f_k = 1 (1 \leq k \leq n)$ iff agent 2 can justify $A_k$; otherwise, $f_k = 0$; $f_B = 1$ iff agent~2 understands $\mathcal{E}$; otherwise, $f_B = 0$. Given a formula $F$ in $\mathcal{E}$, we use $\mathcal{F}_2(\mathcal{E},F)$ to extract the value in $\mathcal{F}_2(\mathcal{E})$ that corresponds to agent 2's feedback on $F$. We write $\mathcal{F}(\mathcal{E})$ and $\mathcal{F}(\mathcal{E},F)$ for short if the agent is clear from the context.
\end{definition}
As we mentioned in the previous section, if agent 2 cannot understand a deduction step, then he cannot understand all the follow-up deduction steps. Thus, if there exists one node in $\mathcal{F}(\mathcal{E})$ that has value 0, all its follow-up nodes towards the root will also have value 0. After hearing the feedback from agent 2, agent 1 can update his beliefs about agent 2 in terms of his justified beliefs. Based on the way in which we define agents' understanding of an explanation, agent 1 can interpret useful information from agent 2's feedback. When agent 2 returns 1 for a hypothesis, it simply means that agent~2 can justify the hypothesis; when agent 2 returns 1 for a derived formula, it means that agent 2 can understand the explanation for the derived formula, which also means that agent~2 can justify the hypotheses as well as the deduction steps in the explanation. On the contrary, when agent 2 returns 0 for a hypothesis, it simply means that agent 2 cannot justify the hypothesis; when agent 2 returns 0 for a derived formula but returns 1 for all of its premises, it means that agent 2 cannot understand the deduction from the premises to the derived formula but can understand all of its premises. In particular, when agent~2 returns 0 for a derived formula as well as some of its premises, agent 1 cannot tell whether the agent can justify the deduction in between. So agent 1 agent will ignore the feedback with respect to this deduction step. Given the above interpretation, agent 1 can remove his possible worlds where opposite information holds. We first define the update of an agent's epistemic state by a truth set (this is the usual definition for announcements~\cite{plaza}), then we define what an explainer agent learns from the explainee's feedback on a given explanation.
\begin{definition}[Update by a Set of Worlds]
Given a multi-agent modular model $\mathcal{M}$, a subset of worlds $X$, and an agent $i$, we define  $\mathcal{M}$ updated with $(i, X)$ by 
\[
\mathcal{M}|(i, X) = (Agt, W', \tilde{R^\prime}, \tilde{*}, \pi^\prime)
\]
as follows:
\begin{itemize}
\item $W' = X$;
\item $R_i^\prime = R_i \cap (X \times X)$;
\item $R_j^\prime = R_j$ for any other agents $j \not= i$;
\item $\pi' = \pi \cap X$.
\end{itemize}
\end{definition}

\begin{definition}[Learning from Feedback]
Given a multi-agent modular model $\mathcal{M}$, the update of $\mathcal{M}$ with agent 1 upon receiving agent 2's feedback on the explanation~$\mathcal{E}$, formally 
$\mathcal{M}|(1,\mathcal{F}(\mathcal{E}))$, is defined by a series of updates:
for each formula $F$ in $\mathcal{E}$ we update $\mathcal{M}$ with $(1,U_F)$ where $U_F$ is given by
\begin{enumerate}
\item if $F \in H(\mathcal{E})$ and $\mathcal{F}(\mathcal{E},F)=1$, then
\[
U_F= \{w \in W \mid \text{there exists a ground term $t$ with 
$\mathcal{M}, w \models \llbracket t \rrbracket_2 F $}\};
\]
\item if $F \in H(\mathcal{E})$ and $\mathcal{F}(\mathcal{E},F)=0$, then
\[
U_F= \{w \in W \mid \text{there is no ground term $t$ with 
$\mathcal{M}, w \models \llbracket t \rrbracket_2 F $}\};
\]
\item if $F \in D(\mathcal{E})$ and $\mathcal{F}(\mathcal{E},F)=1$, then
\[
U_F= \{w \in W \mid \text{there exists a ground term $t$ with $\mathcal{M}, w \models \llbracket t \rrbracket_2 (\Bigarrow{Pr(\mathcal{E}, F)}  F )$}\};
\]
\item if $F \in D(\mathcal{E})$, $\mathcal{F}(\mathcal{E},F)=0$ and for all $P \in Pr(\mathcal{E}, F)$ it is the case that $\mathcal{F}(\mathcal{E},F)=1$, then
\[
U_F= \{w \in W \mid \text{there is no ground term $t$ with $\mathcal{M}, w \models \llbracket t \rrbracket_2 (\Bigarrow{Pr(\mathcal{E}, F)}  F )$}\}.
\]
\end{enumerate}
\end{definition}
Observe that the sets $U_F$ correspond to the characterization of understanding an explanation given in Proposition~\ref{prop:understanding:1}. That means upon receiving feedback, an explainer agent updates in belief on whether the explainee agent understood the given explanation.
%
%
%

\begin{proposition}
Given a multi-agent modular model $\mathcal{M}$, the updated model $\mathcal{M}|(1,\mathcal{F}(\mathcal{E}))$ with agent 1 upon receiving agent 2's feedback on the explanation~$\mathcal{E}$ is still a multi-agent modular model.
\end{proposition}
We then extend our language $\logic{J}$ with a new formula of the form $[j: \mathcal{F}(\mathcal{E})]F$, read as ``$F$ is true after agent $j$ hears feedback $\mathcal{F}(\mathcal{E})$''. We set
\[
\mathcal{M}, w \models [j: \mathcal{F}(\mathcal{E})]F \quad\text{ iff }\quad \mathcal{M}|(j, w, \mathcal{F}(\mathcal{E})), w \models F.
\]
This formula allows to express agent $j$'s updated epistemic state after hearing feedback on an explanation. Note that we assume that the explainee's feedback is truthful (otherwise it could happen that the updated model does not contain the actual world anymore, in which case we would need a truth definition similar to public announcement logic).

\begin{proposition}
Given a multi-agent modular model $\mathcal{M}$ and a world $w$, agent 1 hears feedback $\mathcal{F}(\mathcal{E})$ from agent 2 on explanation $\mathcal{E}$ in world $w$, for any formula $F$ in $\mathcal{E}$,
\begin{itemize}
\item if $F \in H(\mathcal{E})$ and $\mathcal{F}(\mathcal{E}, F) = 1$, then 
\[\mathcal{M}, w \models [1: \mathcal{F}(\mathcal{E})] \Box_1 \triangle_2 F\]
\item if $F \in H(\mathcal{E})$ and $\mathcal{F}(\mathcal{E}, F) = 0$, then
\[\mathcal{M}, w \models [1: \mathcal{F}(\mathcal{E})] \Box_1 \lnot \triangle_2 F\]
\item if $F \in D(\mathcal{E})$ and $\mathcal{F}(\mathcal{E}, F) = 1$, then
    \[\mathcal{M}, w \models [1: \mathcal{F}(\mathcal{E})] \Box_1 \triangle_2 (\Bigarrow{Pr(\mathcal{E}, F)}  F).\]
\item if $F \in D(\mathcal{E})$, $\mathcal{F}(\mathcal{E}, F) = 0$ and for any $P \in Pr(\mathcal{E},F)$ it is the case that $\mathcal{F}(\mathcal{E}, F) = 1$, then
\[\mathcal{M}, w \models [1: \mathcal{F}(\mathcal{E})] \Box_1 \lnot \triangle_2 (\Bigarrow{Pr(\mathcal{E}, F)}  F).\]
\end{itemize}
\end{proposition}

Using the background information about agent 2, agent 1 can further explain what agent 2 cannot justify. However, agent~1 should always remember that its goal is to answer the initial question from agent 2 instead of infinitely explaining what agent 2 cannot justify in the previous explanation. All of these require agent 1 to memorize the conversation with agent 2. We first define the notion of conversation histories between two agents, which always start with a question.
\begin{definition}[Conversation Histories]
A finite conversation history between agent 1 and agent 2 for explaining a propositional formula $F$ is of the form 
\[
\eta = (1: ?F)(2: \mathcal{E}_1)(1: \mathcal{F}(\mathcal{E}_1)) \cdots (2: \mathcal{E}_n)(1: \mathcal{F}(\mathcal{E}_n)),
\]
where $\mathcal{E}_1 \cdots \mathcal{E}_n$ are explanations made by agent 1 to agent 2. We use 
\[
pre(\eta,k) = (1: ?F)(2: \mathcal{E}_1)(1: \mathcal{F}(\mathcal{E}_1)) \cdots (2: \mathcal{E}_k)(1: \mathcal{F}(\mathcal{E}_k))
\]
to denote the prefix of $\eta$ up to the k\textsuperscript{th} position.
\end{definition}
Given a conversation history, agent 1 can decide the explanation to be announced. The decision-making is formalized as a function $\mathcal{E}^*(\eta)$ that takes a conservation history as an input. Basically, agent 1 needs to evaluate whether it is more worthy to further explain what agent 2 cannot justify in the previous explanation, or find another way to explain the initial question, given the current information about agent 2. Let us first use $why(\mathcal{F}(\mathcal{E}))$ to denote the set of formulas that are supposed to have further explanation given feedback $\mathcal{F}(\mathcal{E})$.
\begin{multline*}
why(\mathcal{F}(\mathcal{E})) = \{P \in \mathcal{E} \mid \mathcal{F}(\mathcal{E},P) = 0 \text{ and } \\
(P \in H(\mathcal{E}) \text{ or } \mathcal{F}(\mathcal{E},G) = 1 \text{ for all } G \in Pr(\mathcal{E}, P)).\}
\end{multline*}
The set $why(\mathcal{F}(\mathcal{E}))$ contains formulas in explanation $\mathcal{E}$ on which agent 2 asks for further explanations and that are either hypotheses or derived formula whose premises are justified by agent 2. The rest of the unjustified formulas explanation $\mathcal{E}$ should not be explained, because it is not clear whether agent 2 cannot justify their premises or deductions. As we mentioned before, agent 1's goal is to answer the initial question from agent 2. Thus, when looking for the most preferred explanations, agent 1 should not only consider the explanations for the questions that were asked right now, but also the explanations for the initial question.
\begin{definition}[Most Preferred Explanations]\label{def:mostpreferred}
Assume that we are given a multi-agent modular model~$\mathcal{M}$, a world $w$, and a conversation history 
\[
\eta = (1: ?F)(2: \mathcal{E}_1)(1: \mathcal{F}(\mathcal{E}_1)) \cdots (2: \mathcal{E}_n)(1: \mathcal{F}(\mathcal{E}_n)).
\]
The set of the most preferred explanations with respect to $\eta$ is given by a function $\mathcal{E}^*(\eta)$, which is defined as follows. 
\begin{align*}
\mathcal{M}^\prime &:= \mathcal{M}|(2: \mathcal{E}_1)|(1: \mathcal{F}(\mathcal{E}_1))| \cdots |(2: \mathcal{E}_n)|(1: \mathcal{F}(\mathcal{E}_n))\\
X &:= \lambda^{\mathcal{M}^\prime,w}(\emptyset,F) \cup \bigcup_{G \in why(\mathcal{F}(2,\mathcal{E}_n))} \lambda^{\mathcal{M}^\prime,w}(Pr(\mathcal{E}_n,G),G) \\
\mathcal{E}^*(\eta) &:= \{\mathcal{E} \in X \mid \text{ for any } \mathcal{E}^\prime \in X 
\text{ it is the case that } \mathcal{E}^\prime \precsim^{\mathcal{M}^\prime, w} \mathcal{E} \}
\end{align*}
\end{definition}
In words, given a conversation history $\eta$, agent 1 evaluates explanations from the set $\lambda^{\mathcal{M}^\prime,w}(F)$, which is the set of available explanations for the initial question regarding~$F$, and $\bigcup_{G \in why(\mathcal{F}(\mathcal{E}_n))} \lambda^{\mathcal{M}^\prime,w}(Pr(\mathcal{E}_n,G), G)$, which is the set of available explanations for further explaining what agent 2 cannot understand in $\mathcal{E}_n$, and chooses the one that is most preferred based on the principles in Definition~\ref{def:principles}. For example, if agent~1 finds that it is too time-consuming to make agent~2 understand $\mathcal{E}_n$, because any further explanations contain too many deduction steps, then agent 1 might prefer explaining the claim $F$ in another simple way, if there exists one. It is also important to stress that the evaluation is made with respect to model $\mathcal{M}^\prime = \mathcal{M}|(2: \mathcal{E}_1)|(1: \mathcal{F}(\mathcal{E}_1))| \cdots |(2: \mathcal{E}_n)|(1: \mathcal{F}(\mathcal{E}_n))$, which means that agent 1 considers the latest information about agent 2 that he can infer from conversation $\eta$.

Since explanations are conducted in a conversational way, it is of great importance for the conversation to terminate. The conversation can terminate due to two reasons: either agent 2 has justified the initial claim and thus does not ask any questions, or agent 1 cannot explain more. Apparently, we would like the first one to occur. The following proposition expresses that our explanation approach ensures this desired property if there exists an explanation that agent 2 understands and agent 1 is aware of.
\begin{proposition}
Given a multi-agent modular model $\mathcal{M}$ and a world $w$, if there exists an explanation $\mathcal{E}$ such that $claim(\mathcal{E}) = F$, agent 2 understands $\mathcal{E}$ and agent 1 is aware of~$\mathcal{E}$ in world $w$, then there exists a conversation history 
\[
\eta = (1: ?F)(2: \mathcal{E}_1)(1: \mathcal{F}(\mathcal{E}_1)) \cdots (2: \mathcal{E}_n)(1: \mathcal{F}(\mathcal{E}_n)), 
\]
where $\mathcal{E}_k \in \mathcal{E}^*(pre(\eta,k-1))$, such that 
\begin{itemize}
    \item $\mathcal{M}, w \models [2: \mathcal{E}_n] \cdots [1: \mathcal{F}(\mathcal{E}_1)][2: \mathcal{E}_1] \llbracket t \rrbracket_2 F$, where $t \in \mathrm{Gt}$ is a derived term that agent 2 constructs with respect to $\mathcal{E}$;
    \item all the nodes in $\mathcal{F}(\mathcal{E}_n)$ have value 1.
\end{itemize}
\end{proposition}
\begin{proof}
We first regard all the explanations that agent 1 is aware of for the claim $F$ as a tree that is rooted at $F$, denoted as $T_F$. Every time agent 2 provides feedback to agent 1's previous explanation, agent 1 can have more information about agent 2 (the formulas and deductions that agent 2 can justify). With this information, agent~2 can remove certain formulas and deductions between formulas from $T_F$ so that the explanations containing formulas and deductions that agent 2 cannot justify are not available for agent 1 to agent 2 anymore. Since agent 2's feedback is always truthful due to the reflexivity of the accessibility relations, agent 1 will not remove formulas and deductions that agent 2 can justify from $T_F$. Thus, if there exists an explanation $\mathcal{E}$ such that $claim(\mathcal{E}) = F$, agent 2 understands $\mathcal{E}$ and agent 1 is aware of $\mathcal{E}$, then $\mathcal{E}$ will always stay in $T_F$. Using the approach defined in Definition~\ref{def:mostpreferred} allows agent 1 to look back at the explanations for justifying $F$ in each round of explanation selection. Recall that we have constraints on the evidence function: terms that can be used to justify a formula F are finite, and the formulas that can be justified by a given term are finite. These two constraints ensure that a formula $F$ has finitely many explanations (if there exists). As $\mathcal{E}$ will always stay in $T_F$ and the explanations that can be used to justify $F$ are finite, $\mathcal{E}$ can be found by agent 1 through $\eta$. Therefore, if there exists an explanation $\mathcal{E}$ such that $claim(\mathcal{E}) = F$, agent 2 understands~$\mathcal{E}$ and agent 1 is aware of $\mathcal{E}$, there exists $\eta$ such that after $\eta$ agent 2 can justify $F$. According to Definition \ref{def:mostpreferred}, agent 1 only explains the formulas on which agent 2 asks for further explanations as well as the initial claim. If some of the nodes in $\mathcal{F}(\mathcal{E}_n)$ have value 0, which means that agent 2 has questions about~$\mathcal{E}_n$, then agent 2 cannot construct a ground derived term for $F$ with respect to $\mathcal{E}$ due to his learning approach in Definition \ref{def:learning}, which contradicts the previous conclusion. So all of the nodes in $\mathcal{F}(\mathcal{E}_n)$ have value 1.
\end{proof}

\begin{example}
We illustrate our approach using the example that was mentioned in the introduction. A user $u$ asks a chatbot $c$ why he should drink more water. Because the chatbot has no information about the user, he randomly announces explanation $\mathcal{E}$ to the user, which is "being sick can lead to fluid loss, so drinking more water helps replenish these losses", formalized as $\mathcal{E} = sick / fluid\_loss / drink\_water$. However, the user replies to the chatbot with $\mathcal{F}(\mathcal{E}) = 1/0/0$, making the chatbot believes that the user can justify that he is sick but cannot justify the deduction from $sick$ to $fluid\_loss$,
\begin{align*}
& \mathcal{M}, w \models [c:\mathcal{F}(\mathcal{E})][u:\mathcal{E}] \Box_c \triangle_u sick, \\  
& \mathcal{M}, w \models [c:\mathcal{F}(\mathcal{E})][u:\mathcal{E}] \Box_c \lnot \triangle_u (sick \to fluid\_loss).
\end{align*}
These express the chatbot's belief about the user's background after the chatbot hears his user's feedback $\mathcal{F}(\mathcal{E})$. Next, the chatbot needs to decide what to explain. Regarding $\mathcal{E}$, the chatbot can further explain why being sick can lead to fluid loss, namely $why(\mathcal{F}(\mathcal{E})) = \{fluid\_loss\}$. Besides, the chatbot can also explain why to drink more water when being sick in another way. Suppose the chatbot recognizes that given his belief about the user's background it is too complicated to explain why being sick can lead to fluid loss, namely $\mathcal{E}^\prime = sick / \cdots / fluid\_loss$, but he can simply tell the user that being sick can make you thirsty, so you should drink more water, formalized as $\mathcal{E}^{\prime\prime} = sick / thirsty / drink\_water$. Given the conversation history 
\[\eta = (c: ?drink\_water)(u: \mathcal{E})(c: \mathcal{F}(\mathcal{E})),\] 
the chatbot believes that the user can justify the deduction from $sick$ to $thirsty$ and from $thirsty$ to $drink\_water$ in $\mathcal{E}^{\prime\prime}$. Using the principles in Definition~\ref{def:principles}, the chatbot specifies the preference over $\mathcal{E}^\prime$ and $\mathcal{E}^{\prime\prime}$, and gets $\mathcal{E}^\prime \prec^{\mathcal{M}^\prime,w} \mathcal{E}^{\prime\prime}$, where $\mathcal{M}^\prime = \mathcal{M} | (u: \mathcal{E}) | (c: \mathcal{F}(\mathcal{E}))$. Thus, the chatbot announces $\mathcal{E}^{\prime\prime}$ to the user. Since the chatbot's belief about the user is always true, the user can justify all parts in $\mathcal{E}^{\prime\prime}$ and thus understand $\mathcal{E}^{\prime\prime}$. More precisely, assume that $t,s,r \in \mathrm{Gt}$, and
\begin{align*}
& \mathcal{M}, w \models [c:\mathcal{F}(\mathcal{E})][u:\mathcal{E}]\llbracket t \rrbracket_u sick,\\
& \mathcal{M}, w \models [c:\mathcal{F}(\mathcal{E})][u:\mathcal{E}]\llbracket s \rrbracket_u (sick \to thirsty),\\
& \mathcal{M}, w \models [c:\mathcal{F}(\mathcal{E})][u:\mathcal{E}]\llbracket r \rrbracket_u (thirsty \to drink\_water).
\end{align*}
The user then constructs term $s \cdot t$ for $thirsty$ and term $r \cdot (s \cdot t)$ for $drink\_water$ with respect to $\mathcal{E}^{\prime\prime}$, and gains more justified beliefs accordingly after hearing $\mathcal{E}^{\prime\prime}$,
\begin{align*}
& \mathcal{M}, w \models [u: \mathcal{E}^{\prime\prime}][c:\mathcal{F}(\mathcal{E})][u:\mathcal{E}] \llbracket s \cdot t \rrbracket_u thirsty,\\
& \mathcal{M}, w \models [u: \mathcal{E}^{\prime\prime}][c:\mathcal{F}(\mathcal{E})][u:\mathcal{E}] \llbracket r \cdot (s \cdot t) \rrbracket_u drink\_water.
\end{align*}
Therefore, the user's feedback on $\mathcal{E}^{\prime\prime}$ is $\mathcal{F}(\mathcal{E}^{\prime\prime}) = 1/1/1$. After the chatbot hears this feedback, the conversation terminates.
\end{example}

\section{Conclusion}
It is important for our AI systems to provide personalized explanations to users that are relevant to them and match their background knowledge. A conversation between explainers and explainees not only allows explainers to obtain the explainees' background but also allows explainees to better understand the explanations. In this paper, we proposed an approach that allows an explainer to communicate personalized explanations to explainee through having consecutive conversations with the explainee. It is built on the idea that the explainee understands an explanation if and only if he can justify all formulas in the explanation. In a conversation for explanations, the explainee provides his feedback on the explanation that has just been announced, while the explainer interprets the explainee's background from the feedback and then selects an explanation for announcement given what has learned about the explainee. We proved that the conversation will terminate due to the explainee's justification of the initial claim as long as there exists an explanation for the initial claim that the explainee understands and the explainer is aware of. In the future, we would like to extend our approach with another dimension for evaluating explanations: the \emph{acceptance} of explanations. The explanation that is selected by the explainer should be not only understood, but also accepted, by the explainee. For example, a policeman does not accept an explanation for over-speed driving due to heading for a party. For this part of study, our framework needs to be enriched with evaluation standards such as values and personal norms. On the technical level, we would like to extend our logic so that the explainer can reason about the explainee's feedback for gaining information about the explainees' background according to our idea about understanding explanations.

\bibliography{sigproc,JLBibliography}
\end{document}